\documentclass[english]{amsart}
\usepackage[T1]{fontenc}
\usepackage[latin9]{inputenc}
\usepackage{color}
\usepackage{verbatim}
\usepackage{prettyref}
\usepackage{float}
\usepackage{amsthm}
\usepackage{amssymb}
\usepackage[numbers]{natbib}

\makeatletter

\floatstyle{ruled}
\newfloat{algorithm}{tbp}{loa}
\providecommand{\algorithmname}{Algorithm}
\floatname{algorithm}{\protect\algorithmname}

 \usepackage{algorithmic}
\theoremstyle{plain}
\newtheorem{thm}{\protect\theoremname}
  \theoremstyle{plain}
  \newtheorem{lem}[thm]{\protect\lemmaname}
  \theoremstyle{definition}
  \newtheorem{defn}[thm]{\protect\definitionname}
  \theoremstyle{definition}
  \newtheorem{example}[thm]{\protect\examplename}
  \theoremstyle{remark}
  \newtheorem{rem}[thm]{\protect\remarkname}

\usepackage{fullpage}\usepackage{amsthm}\usepackage{stmaryrd}

\@ifundefined{definecolor}
 {\usepackage{color}}{}
\usepackage{refstyle}\usepackage{float}\usepackage{amsthm}

\usepackage{a4wide}\ifx\pdfpageheight\undefined
   \usepackage[dvips,colorlinks=true,linkcolor=blue,citecolor=red,
      urlcolor=green]{hyperref}\makeatletter
   \edef\Gin@extensions{\Gin@extensions,.mps}
   
   \makeatother
\else
   \usepackage[bookmarksopen=false,pdftex=true,breaklinks=true,
       hyperindex=true,pdfstartview=FitH,colorlinks=true,
      pdfpagelabels=true,colorlinks=true,linkcolor=blue,
      citecolor=red,urlcolor=green,hypertexnames=false
      ]{hyperref}\fi

\usepackage{algorithm}\usepackage{algorithmic}\floatstyle{ruled}
\newfloat{algorithm}{tbp}{loa}
\floatname{algorithm}{Algorithm}
\usepackage{algorithmic}

\newcommand{\bigO}{{O\tilde{\phantom{\imath}}}}

\newcommand{\cdeg}{{\rm cdeg}\,}

\newcommand{\diag}{{\rm diag}\,}

\DeclareMathOperator{\determinant}{determinant}

\DeclareMathOperator{\colBasis}{ColumnBasis}

\makeatother

\usepackage{babel}
  \providecommand{\definitionname}{Definition}
  \providecommand{\examplename}{Example}
  \providecommand{\lemmaname}{Lemma}
  \providecommand{\remarkname}{Remark}
\providecommand{\theoremname}{Theorem}

\begin{document}

\title{Fast and deterministic computation of the determinant of a polynomial matrix}

\author{Wei Zhou and George Labahn}

\thanks{Cheriton School of Computer Science, University of Waterloo, Waterloo
ON, Canada N2L 3G1 \textbf{\{w2zhou,glabahn\}@uwaterloo.ca}}
\begin{abstract}
Given a square, nonsingular matrix of univariate polynomials 
$\mathbf{F}\in\mathbb{K}[x]^{n\times n}$ over a field $\mathbb{K}$, we give a
deterministic algorithm for finding the determinant of $\mathbf{F}$. The complexity of the algorithm is $\bigO \left(n^{\omega}s\right)$ field operations where $s$ is the average column degree or the average row degree of $\mathbf{F}$. Here $\bigO$ notation is Big-$O$ with log factors omitted and $\omega$ 
is the exponent of matrix multiplication. 
\end{abstract}
\maketitle

\section{Introduction}

Let $\mathbf{F}\in\mathbb{K}[x]^{n\times n}$ be a square, nonsingular polynomial matrix with
$\mathbb{K}$ a field. In this paper we give a deterministic algorithm for finding the determinant of $\mathbf{F}$. The complexity of the algorithm is $\bigO\left(n^{\omega}s\right)$ field operations from $\mathbb{K}$ where $s$ is the average column degree or the average row degree of $\mathbf{F}$. Here $\bigO$ denotes $O$ with $\log^c (nd)$ factors suppressed for some positive real constant $c$ and $\omega$ is the exponent of matrix multiplication. 
The fact that the complexity of determinant computation is related to the complexity of matrix multiplication is well-known. In the case of matrices over a field, for example, Bunch and Hopcroft \cite{BunchHopcroft1974} showed that if there exists a fast algorithm for matrix multiplication then there also exists an algorithm for determinant computation with the same exponent. 

In the case of square matrices of polynomials of degree at most $d$, Storjohann \cite{storjohann:phd2000} gives a recursive deterministic algorithm to compute a determinant making use of fraction-free Gaussian elimination with a cost of $\bigO(n^{\omega + 1} d)$ operations. A $O(n^3 d^2)$ deterministic algorithm was later given by Mulders and Storjohann \cite{mulders-storjohann:2003}, modifying their weak Popov form computation. Using low rank perturbations, Eberly, et al  \cite{EberlyGiesbrechtVillard} gave a determinant algorithm requiring  $\bigO(n^{2 + \omega/2} d)$ field operations, while Storjohann \cite{storjohann:2003} used higher order lifting to give an algorithm which reduces the complexity to $(n^\omega ( \log n)^2 d^{1 + \epsilon})$ field operations. Finally, we mention the algorithm of Giorgi et al  \citep{Giorgi2003} which computes the determinant with complexity $O^{\sim}\left(n^{\omega}d\right)$. However the algorithms in both  \cite{EberlyGiesbrechtVillard} and \cite{storjohann:2003} are both probabilistic, while the algorithm from  \cite{Giorgi2003} only works efficiently on a class of generic input matrices, matrices that are well behaved in the computation.  
Deterministic algorithms for general polynomial matrices with complexity similar to that of fast matrix multiplication have not appeared previously. 

In the case of an arbitrary commutative ring (with multiplicative unit) or integers other fast determinant algorithms have been given by  Kaltofen \cite{kaltofen92}, 
Abbott et al \cite{abbott}, Eberly et al  \cite{EberlyGiesbrechtVillard} and Kaltofen and Villard \cite{KaltofenVillard}. We refer the reader to the last named  paper and the references therein for more details on efficient determinant computation of such matrices.

Our algorithm takes advantage of a fast algorithm  \cite{za2012} for computing a shifted minimal kernel basis to efficiently eliminate blocks of a polynomial matrix. More specifically, we use kernel bases to partition our input $\mathbf{F}$ as
$$
\mathbf{F} \cdot \mathbf{U}=\begin{bmatrix}\mathbf{G}_{1} & 0\\
* & \mathbf{G}_{2}
\end{bmatrix}
$$
with $\mathbf{U}$ unimodular.  Such a unimodular transformation almost preserves the determinant of  $\mathbf{F}$, but results in  an extra factor coming from the determinant of the unimodular matrix $\mathbf{U}$, a nonzero field element in $\mathbb{K}$. The computation of the determinant of $\mathbf{F}$ can therefore be reduced to the computations of the determinants of $\mathbf{U}$,  $\mathbf{G}_{1}$ and $\det \mathbf{G}_{2}$. The computations of $\det \mathbf{G}_{1}$ and $\det \mathbf{G}_{2}$ are similar to the original problem of computing $\det \mathbf{F}$, but with input matrices of lower dimension and possibly higher degrees.  To achieve the desired efficiency, however, these computations need to be done without actually determining the unimodular matrix $\mathbf{U}$, since its potential large degree size may 
prevent it from being efficiently computed.  We show how the determinant of $\mathbf{U}$ can be computed without actually computing the entire unimodular multiplier $\mathbf{U}$. In addition, for fast, recursive computation, the degrees of each of the diagonal blocks need to be controlled in order to ensure that these are also not too large. 
We accomplish this by making use of the concepts of {\em shifted} minimal kernel bases and column bases of polynomial matrices. Shifts basically allow us to control the computations using column degrees rather than the degree of the polynomial matrix. This becomes an issue when the degrees of the input columns vary considerably from column to column (and hence to the degree of the input). The shifted kernel and column bases computations can be done efficiently using algorithms from \cite{za2012}  and \cite{za2013}.  
We remark that the use of shifted minimal kernel bases and column bases, used in the context of fast block elimination, have also been used for deterministic algorithms for inversion \cite{zls2014} and unimodular completion
\cite{zl2014} of polynomial matrices.

The remainder of this paper is organized as follows. In the next section we give preliminary information for shifted degrees, kernel and column bases of polynomial matrices. Section 3 then contains the algorithm for recursively computing the diagonal elements of a triangular form and a method to compute the determinants of the unimodular matrices.
The paper ends with a conclusion and topics for future research. 

\section{Preliminaries}

In this section we 
give the basic definitions and properties of {\em shifted
degree}, {\em minimal kernel basis} and {\em column basis} for a matrix
of polynomials. These will be the building blocks used in our algorithm.

\subsection{Shifted Degrees}

Our methods makes use of the concept of {\em shifted} degrees of
polynomial matrices \citep{BLV:1999}, basically shifting the importance
of the degrees in some of the rows of a basis. For a column vector
$\mathbf{p}=\left[p_{1},\dots,p_{n}\right]^{T}$ of univariate polynomials
over a field $\mathbb{K}$, its column degree, denoted by $\cdeg\mathbf{p}$,
is the maximum of the degrees of the entries of $\mathbf{p}$, that
is, 
\[
\cdeg~\mathbf{p}=\max_{1\le i\le n}\deg p_{i}.
\]
 The \emph{shifted column degree} generalizes this standard column
degree by taking the maximum after shifting the degrees by a given
integer vector that is known as a \emph{shift}. More specifically,
the shifted column degree of $\mathbf{p}$ with respect to a shift
$\vec{s}=\left[s_{1},\dots,s_{n}\right]\in\mathbb{Z}^{n}$, or the
\emph{$\vec{s}$-column degree} of $\mathbf{p}$ is 
\[
\cdeg_{\vec{s}}~\mathbf{p}=\max_{1\le i\le n}[\deg p_{i}+s_{i}]=\deg(x^{\vec{s}}\cdot\mathbf{p}),
\]
 where 
\[
x^{\vec{s}}=\diag\left(x^{s_{1}},x^{s_{2}},\dots,x^{s_{n}}\right)~.
\]
 For a matrix $\mathbf{P}$, we use $\cdeg\mathbf{P}$ and $\cdeg_{\vec{s}}\mathbf{P}$
to denote respectively the list of its column degrees and the list
of its shifted $\vec{s}$-column degrees. When $\vec{s}=\left[0,\dots,0\right]$,
the shifted column degree specializes to the standard column degree.

 Shifted degrees have been used previously in polynomial matrix computations
and in generalizations of some matrix normal forms \citep{BLV:jsc06}.
The shifted column degree is equivalent to the notion of \emph{defect}
commonly used in the literature.

Along with shifted degrees we also make use of the notion of a matrix
polynomial being column (or row) reduced. A polynomial matrix $\mathbf{F}$
is column reduced if the leading column coefficient matrix, that is
the matrix 
\[
[\mbox{coeff}(f_{ },x,d_{j})]_{1\leq i,j\leq n},\mbox{ with }\vec{d}=\mbox{cdeg }\mathbf{F},
\]
 has full rank. A polynomial matrix $\mathbf{F}$ is $\vec{s}$-column
reduced if $x^{\vec{s}}\mathbf{F}$ is column reduced. 

The usefulness of the shifted degrees can be seen from their applications
in polynomial matrix computation problems \citep{ZL2012,zl2014,za2012,zls2014}.
One of its uses is illustrated by the following 
lemma, which follows directly from the definition of shifted degree.
\begin{lem}
\label{lem:predictableDegree} Let $\vec{s}$ be a shift whose entries
bound the corresponding column degrees of $\mathbf{A} \in
\mathbb{K}^{\ast \times m}$. Then for any polynomial matrix
$\mathbf{B}\in\mathbb{K}\left[x\right]^{m\times \ast}$, the column
degrees of $\mathbf{A}\cdot \mathbf{B}$ are bounded by the corresponding
$\vec{s}$-column degrees of $\mathbf{B}$.
\end{lem}
An essential subroutine needed in our algorithm, also based on the
use of the shifted degrees, is the efficient multiplication of a
pair of matrices $\mathbf{A} \cdot \mathbf{B}$ with unbalanced
degrees.  The following result follows as a special case of~\citep[Theorem
5.6]{zhou:phd2012}.  The notation $\sum\vec{s}$, for any list
$\vec{s}$, denotes the sum of all entries in $\vec{s}$.

\begin{thm}
\label{thm:multiplyUnbalancedMatrices} Let
$\mathbf{A}\in\mathbb{K}[x]^{n \times m}$ and
$\mathbf{B}\in\mathbb{K}[x]^{m\times n}$ be given, $m\le n$. 
Suppose
$\vec{s} \in \mathbb{Z}_{\geq 0}^n$ is a shift that bounds the
corresponding column degrees of $\mathbf{A}$, and $\sum \vec{s} \ge \sum \cdeg_{\vec{s}} \mathbf{B}$. Then the 
product $\mathbf{A}\cdot \mathbf{B}$ can be computed in $O^{\sim}\left(n^{\omega}s\right)$ field operations from $\mathbb{K}$, where $s=\sum\vec{s}/n$ is the average of the entries of $\vec{s}$.
\end{thm}

\subsection{Kernel and Column Bases}

The kernel of $\mathbf{F}\in\mathbb{K}\left[x\right]^{m\times n}$
is the $\mathbb{F}\left[x\right]$-module 
\[
\left\{ \mathbf{p}\in\mathbb{K}\left[x\right]^{n}~|~\mathbf{F}\mathbf{p}=0\right\} 
\]
 with a kernel basis of $\mathbf{F}$ being a basis of this module.
Formally, we have:
\begin{defn}
\label{def:kernelBasis}Given $\mathbf{F}\in\mathbb{K}\left[x\right]^{m\times n}$,
a polynomial matrix $\mathbf{N}\in\mathbb{K}\left[x\right]^{n\times k}$
is a (right) kernel basis of $\mathbf{F}$ if the following properties
hold: 
\begin{enumerate}
\item $\mathbf{N}$ is full-rank. 
\item $\mathbf{N}$ satisfies $\mathbf{F}\cdot\mathbf{N}=0$. 
\item Any $\mathbf{q}\in\mathbb{K}\left[x\right]^{n}$ satisfying $\mathbf{F}\mathbf{q}=0$
can be expressed as a linear combination of the columns of $\mathbf{N}$,
that is, there exists some polynomial vector $\mathbf{p}$ such that
$\mathbf{q}=\mathbf{N}\mathbf{p}$. 
\end{enumerate}
\end{defn}
It is not difficult to show that any pair of kernel bases $\mathbf{N}$
and $\mathbf{M}$ of $\mathbf{F}$ 
are unimodularly equivalent, that is, $\mathbf{N} = \mathbf{M} \cdot \mathbf{V}$ for some unimodular matrix $\mathbf{V}$.

A $\vec{s}$-minimal kernel basis of $\mathbf{F}$ is just a kernel
basis that is $\vec{s}$-column reduced. 
\begin{defn}
Given $\mathbf{F}\in\mathbb{K}\left[x\right]^{m\times n}$, a polynomial
matrix $\mathbf{N}\in\mathbb{K}\left[x\right]^{n\times k}$ is a $\vec{s}$-minimal
(right) kernel basis of $\mathbf{F}$ if\textbf{ $\mathbf{N}$} is
a kernel basis of $\mathbf{F}$ and $\mathbf{N}$ is $\vec{s}$-column
reduced. We also call a $\vec{s}$-minimal (right) kernel basis of
$\mathbf{F}$ a $\left(\mathbf{F},\vec{s}\right)$-kernel basis.


\end{defn}


We will need the following result from \citep{za2012} to bound the
sizes of kernel bases.
\begin{thm}
\label{thm:boundOfSumOfShiftedDegreesOfKernelBasis}Suppose $\mathbf{F}\in\mathbb{K}\left[x\right]^{m\times n}$
and $\vec{s}\in\mathbb{Z}_{\ge0}^{n}$ is a shift with entries bounding
the corresponding column degrees of $\mathbf{F}$. Then the sum of
the $\vec{s}$-column degrees of any $\vec{s}$-minimal kernel basis
of $\mathbf{F}$ is bounded by $\sum\vec{s}$. 
\end{thm}
A column basis of $\mathbf{F}$ is a basis for the $\mathbb{K}\left[x\right]$-module
\[
\left\{ \mathbf{F}\mathbf{p}~|~\mathbf{p}\in\mathbb{K}\left[x\right]^{n}~\right\} ~.
\]
 Such a basis can be represented as a full rank matrix $\mathbf{T}\in\mathbb{K}\left[x\right]^{m\times r}$
whose columns are the basis elements. A column basis is not unique
and indeed any column basis right multiplied by a unimodular polynomial
matrix gives another column basis.

The cost of kernel basis computation is given in \citep{za2012} while
the cost of column basis computation is given in \citep{za2013}.
In both cases they make heavy use of fast methods for order bases
(also sometimes referred to as sigma bases or minimal approximant bases) \citep{BeLa94,Giorgi2003,ZL2012}.
\begin{thm}
\label{thm:costGeneral} Let $\mathbf{F}\in\mathbb{K}\left[x\right]^{m\times n}$
with $\vec{s}=\cdeg\mathbf{F}$. Then a $\left(\mathbf{F},\vec{s}\right)$-kernel
basis can be computed with a cost of $\bigO\left(n^{\omega}s\right)$ field
operations where $s=\sum\vec{s}/n$ is the average column degree of
$\mathbf{F}$. 
\end{thm}
~
\begin{thm}
\label{thm:fastcolbasis} 
The algorithm from \citep{za2013} can compute a column basis of a polynomial matrix $\mathbf{F}$ deterministically with
$\bigO\left(nm^{\omega-1}s\right)$ field operations
in $\mathbb{K}$, where $s$ is the average average column degree
of $\mathbf{F}$. In addition, if $r$ is the rank of $\mathbf{F}$, then 
the column basis computed has column
degrees bounded by the $r$ largest column degrees of $\mathbf{F}$,
\end{thm}

\begin{example}
\label{ex:example1} Let
\[
\mathbf{F}=\left[\begin{array}{rcrcr}
x & -{x}^{3} & -2\,{x}^{4} & 2x & -{x}^{2}\\
\noalign{\medskip}1 & -1 & -2\, x & 2 & -x\\
\noalign{\medskip}-3 & 3\,{x}^{2}+x & 2\,{x}^{2} & -\,{x}^{4}+1 & 3\, x
\end{array}\right]
\]
 be a $3\times5$ matrix over $\mathbb{Z}_{7}[x]$ having column degree
$\vec{s}=(1,3,4,4,2)$. Then a column space $\mathbf{G}$ and a
kernel basis $\mathbf{N}$ of $\mathbf{F}$ are given by 
\[
\mathbf{G}=\left[\begin{array}{rcr}
x & -{x}^{3} & -2\,{x}^{4}\\
\noalign{\medskip}1 & -1 & -2\, x\\
\noalign{\medskip}-3 & 3\,{x}^{2}+x & 2\,{x}^{2}
\end{array}\right]~~\mbox{ and }~~\mathbf{N}:=\left[\begin{array}{rc}
-1 & x\\
\noalign{\medskip}-{x}^{2} & 0\\
\noalign{\medskip}-3\, x & 0\\
\noalign{\medskip}-3 & 0\\
\noalign{\medskip}0 & 1
\end{array}\right]~.
\]
 For example, if $\{\mathbf{g}_{i}\}_{i=1,...,5}$ denote the columns
of $\mathbf{G}$ then column $4$ of $\mathbf{F}$ - denoted by $\mathbf{f}_{4}$
- is given by 
\[
\mathbf{f}_{4}=-2~\mathbf{g}_{1}-2x^{2}~\mathbf{g}_{2}+x~\mathbf{g}_{3}+2~\mathbf{g}_{4}.
\]
 Here $\cdeg_{\vec{s}}\mathbf{N}=(5,2)$ with shifted leading coefficient
matrix 
\[
\mbox{lcoeff}_{\vec{s}}(\mathbf{N})=\left[\begin{array}{rc}
0 & 1\\
-1 & 0\\
-3 & 0\\
0 & 0\\
0 & 1
\end{array}\right].
\]
The kernel basis  $\mathbf{N}$  satisfies  $\mathbf{F} \cdot \mathbf{N} = \mathbf{0}$. 
 Since $\mbox{ lcoeff}_{\vec{s}}(\mathbf{N})$ has full rank, it is $\vec{s}$-column reduced, and we have
that $\mathbf{N}$ is a $\vec{s}$-minimal kernel basis. \qed 
\end{example}
Column bases and kernel bases are closely related, as shown by the
following result from \cite{zhou:phd2012,za2013}.
\begin{lem}
\label{lem:unimodular_kernel_columnBasis} Let $\mathbf{F}\in\mathbb{K}\left[x\right]^{m\times n}$
and suppose $\mathbf{U}\in\mathbb{K}\left[x\right]^{n\times n}$ is
a unimodular matrix such that $\mathbf{F} \cdot \mathbf{U}=\left[0,\mathbf{T}\right]$
with $\mathbf{T}$ of full column rank. Partition $\mathbf{U}=\left[\mathbf{U}_{L},\mathbf{U}_{R}\right]$
such that $\mathbf{F}\cdot\mathbf{U}_{L}=0$ and $\mathbf{F} \cdot \mathbf{U}_{R}=\mathbf{T}$.
Then 
\begin{enumerate}
\item $\mathbf{U}_{L}$ is a kernel basis of $\mathbf{F}$ and $\mathbf{T}$
is a column basis of~$\mathbf{F}$. 
\item If $\mathbf{N}$ is any other kernel basis of $\mathbf{F}$, then
$\mathbf{U}^{*}=\left[\mathbf{N},~\mathbf{U}_{R}\right]$ is
unimodular and also unimodularly transforms $\mathbf{F}$ to $\left[0,\mathbf{T}\right]$. 
\end{enumerate}
\end{lem}

\section{\label{sec:diagonals}Recursive Computation}

In this section we show how to recursively compute the determinant
of a nonsingular input matrix $\mathbf{F}\in\mathbb{K}\left[x\right]^{n\times n}$
having column degrees $\vec{s}$. 
The computation makes use of fast kernel basis and column basis computation.

Consider unimodularly transforming $\mathbf{F}$ to 
\begin{equation}
\mathbf{F} \cdot \mathbf{U}=\mathbf{G}=\begin{bmatrix}\mathbf{G}_{1} & 0\\
* & \mathbf{G}_{2}
\end{bmatrix},\label{eq:step1HermiteDiagonal}
\end{equation}
which eliminates a top right block and gives two square diagonal blocks
$\mathbf{G}_{1}$ and\textbf{ $\mathbf{G}_{2}$} in $\mathbf{G}$.
Then the determinant of $\mathbf{F}$ can be computed as 
\begin{equation}
\det\mathbf{F}=\frac{\det\mathbf{G}}{\det\mathbf{U}}=\frac{\det\mathbf{G}_{1}\cdot\det\mathbf{G}_{2}}{\det\mathbf{U}},\label{eq:determinantFromDiagonalBlocks}
\end{equation}
which requires us to first compute $\det\mathbf{G}_{1}$, $\det\mathbf{G}_{2}$,
and $\det\mathbf{U}$. The same procedure can then be applied 
to compute the determinant of $\mathbf{G}_{1}$ and the determinant
of $\mathbf{G}_{2}$. This can be repeated recursively until the dimension
becomes $1$. 

One major obstacle of this approach, however, is that the degrees
of the unimodular matrix $\mathbf{U}$ and the matrix $\mathbf{G}$
can be too large for efficient computation. %
{} 
To sidestep this issue, we will show that the matrices $\mathbf{G}_{1}$, 
$\mathbf{G}_{2}$, and the scalar $\det\mathbf{U}$ can in fact be computed
without computing the entire matrices $\mathbf{G}$ and $\mathbf{U}$.

\subsection{Computing the diagonal blocks}

Suppose we want $\mathbf{G}_{1}$ to have dimension $k$. We can partition\textbf{
$\mathbf{F}$ }as $\mathbf{F}=\begin{bmatrix}\mathbf{F}_{U}\\
\mathbf{F}_{D}
\end{bmatrix}$ with $k$ rows in $\mathbf{F}_{U}$, and note that both $\mathbf{F}_{U}$
and $\mathbf{F}_{D}$ are of full-rank since $\mathbf{F}$ is assumed
to be nonsingular. By partitioning $\mathbf{U}=\begin{bmatrix}\mathbf{U}_{L}~~, & \mathbf{U}_{R}\end{bmatrix}$,
with $k$ columns in $\mathbf{U}_{L}$, then 
\begin{equation}
\mathbf{F} \cdot \mathbf{U}=\begin{bmatrix}\mathbf{F}_{U}\\
\mathbf{F}_{D}
\end{bmatrix}\begin{bmatrix}\mathbf{U}_{L} & \mathbf{U}_{R}\end{bmatrix}=\begin{bmatrix}\mathbf{G}_{1} & 0\\
* & \mathbf{G}_{2}
\end{bmatrix}=\mathbf{G}.\label{eq:UPartitioned}
\end{equation}
Notice that the matrix $\mathbf{G}_{1}$ is nonsingular and is therefore
a column basis of $\mathbf{F}_{U}$. As such this can be efficiently
computed 
as mentioned in Theorem \ref{thm:fastcolbasis}. In addition, the
column basis algorithm makes the resulting column degrees of $\mathbf{G}_{1}$
small enough for $\mathbf{G}_{1}$ to be efficiently used again as
the input matrix of a new subproblem in the recursive procedure.
\begin{lem}
\label{lem:firstDiagonalBlock}The first diagonal block $\mathbf{G}_{1}$
in $\mathbf{G}$ can be computed with a cost of $\bigO\left(n^{\omega}s\right)$
and with column degrees bounded by the $k$ largest column degrees
of $\mathbf{F}$.
\end{lem}
For computing the second diagonal block $\mathbf{G}_{2}$, notice
that we do not need a complete unimodular matrix $\mathbf{U}$, as only
$\mathbf{U}_{R}$ is needed to compute $\mathbf{G}_{2}=\mathbf{F}_{D}\mathbf{U}_{R}$.
In fact, \prettyref{lem:unimodular_kernel_columnBasis} tells us much
more. It tells us that the matrix $\mathbf{U}_{R}$ is a right kernel
basis of $\mathbf{F}$, which makes the top right block of $\mathbf{G}$
zero. In addition the kernel basis $\mathbf{U}_{R}$ can be replaced by any
other kernel basis of $\mathbf{F}$ to give another unimodular matrix
that also transforms $\mathbf{F}_{U}$ to a column basis and also
eliminates the top right block of $\mathbf{G}$. 
\begin{lem}
\label{lem:oneStepHermiteDiagonal} Partition $\mathbf{F}=\begin{bmatrix}\mathbf{F}_{U}\\
\mathbf{F}_{D}
\end{bmatrix}$ and suppose $\mathbf{G}_{1}$ is a column basis of $\mathbf{F}_{U}$
and $\mathbf{N}$ a kernel basis of $\mathbf{F}_{U}$. Then there
is a unimodular matrix $\mathbf{U}=\left[~*~,~\mathbf{N}\right]$
such that 
\[
\mathbf{F} \cdot \mathbf{U}=\begin{bmatrix}\mathbf{G}_{1} & 0\\
* & \mathbf{G}_{2}
\end{bmatrix},
\]
 where $\mathbf{G}_{2}=\mathbf{F}_{D} \cdot \mathbf{N}$. If $\mathbf{F}$
is square nonsingular, then $\mathbf{G}_{1}$ and $\mathbf{G}_{2}$
are also square nonsingular. 
\end{lem}
Note that the blocks represented by the symbol $*$ are not needed in
our computation. These blocks may have very large degrees and cannot
be computed efficiently. 

We have just seen how $\mathbf{G}_{1}$
and $\mathbf{G}_{2}$ can be determined without computing the unimodular matrix $\mathbf{U}$. We still need to make sure that
$\mathbf{G}_{2}$ can be computed efficiently, which can be done by using the existing
algorithms for kernel basis computation and the multiplication of
matrices with unbalanced degrees. We also require that the column
degrees of $\mathbf{G}_{2}$ be small enough for $\mathbf{G}_{2}$
to be efficiently used again as the input matrix of a new subproblem
in the recursive procedure.
\begin{lem}
\label{lem:secondDiagonalBlock}
The second diagonal
block $\mathbf{G}_{2}$ can be computed with a cost of $\bigO\left(n^{\omega}s\right)$
field operations. Furthermore $\sum\cdeg\mathbf{G}_{2}\le\sum\vec{s}$.
\end{lem}
\begin{proof}
From \prettyref{lem:oneStepHermiteDiagonal} we have that $\mathbf{G}_{2}=\mathbf{F}_{D} \cdot \mathbf{N}$ 
with $\mathbf{N}$ a kernel basis  of $\mathbf{F}_{U}$. In fact, this
kernel basis can be made $\vec{s}$-minimal using the algorithm from
\cite{za2012}, and computing such a $\vec{s}$-minimal kernel basis
of $\mathbf{F}_{U}$ costs $\bigO\left(n^{\omega}s\right)$ field
operations by \prettyref{thm:costGeneral}. 
In addition, from \prettyref{thm:boundOfSumOfShiftedDegreesOfKernelBasis} the sum of
the $\vec{s}$-column degrees of such a $\vec{s}$-minimal $\mathbf{N}$
is bounded by $\sum\vec{s}$.

For the matrix multiplication $\mathbf{F}_{D} \cdot \mathbf{N}$, the sum
of the column degrees of $\mathbf{F}_{D}$ and the sum of the $\vec{s}$-column
degrees of $\mathbf{N}$ are both bounded by $\sum\vec{s}$. 
Therefore 
we can 
apply \prettyref{thm:multiplyUnbalancedMatrices}
directly to multiply $\mathbf{F}_{D}$ and $\mathbf{N}$ with a cost
of $\bigO\left(n^{\omega}s\right)$ field operations. 

The second statement follows from \prettyref{lem:predictableDegree}.
\end{proof}

The computation of a kernel basis $\mathbf{N}$ of $\mathbf{F}_{U}$ is actually also used as an intermediate step by the column basis algorithm for computing the column basis $\mathbf{G}_{1}$  \citep{za2013}. In other words, we can get this kernel basis from the computation of $\mathbf{G}_{1}$ with no additional work.


\subsection{Determinant of the unimodular matrix}

\prettyref{lem:firstDiagonalBlock} and \prettyref{lem:secondDiagonalBlock}
show that the two diagonal blocks in \prettyref{eq:step1HermiteDiagonal}
can be computed efficiently. In order to compute the determinant of $\mathbf{F}$
using \prettyref{eq:determinantFromDiagonalBlocks}, we still need
to know the determinant of the unimodular matrix $\mathbf{U}$
satisfying \prettyref{eq:UPartitioned}, or equivalently, we can also
find out the determinant of $\mathbf{V}=\mathbf{U}^{-1}$. The column basis computation from  \cite{za2013} for computing the diagonal block $\mathbf{G}_{1}$  also gives $\mathbf{U}_{R}$,
the matrix consisting of the right $(n-k)$ columns of $\mathbf{U}$, which is also a right kernel
basis of $\mathbf{F}_{U}$. In fact, this column basis computation also gives
 a right factor multiplied with the column basis $\mathbf{G}_{1}$ to give  $\mathbf{F}_{U}$.  The following lemma shows that this right factor coincides with the the matrix $\mathbf{V}_{U}$ consisting of the top $k$ rows of $\mathbf{V}$.  The column basis computation therefore gives both $\mathbf{U}_{R}$ and $\mathbf{V}_{U}$ with no additional work.
\begin{lem}
Let $k$ be the dimension of $\mathbf{G}_{1}$. The matrix $\mathbf{V}_{U}\in\mathbb{K}\left[x\right]^{k\times n}$
satisfies $\mathbf{G}_{1} \cdot \mathbf{V}_{U}=\mathbf{F}_{U}$ if and only
if $\mathbf{V}_{U}$ is the submatrix of $\mathbf{V}=\mathbf{U}^{-1}$
consisting of the top $k$ rows of $\mathbf{V}$. \end{lem}
\begin{proof}
The proof follows directly from 
\[
\mathbf{G} \cdot \mathbf{V}=\begin{bmatrix}\mathbf{G}_{1} & 0\\
* & \mathbf{G}_{2}
\end{bmatrix}\begin{bmatrix}\mathbf{V}_{U}\\
\mathbf{V}_{D}
\end{bmatrix}=\begin{bmatrix}\mathbf{F}_{U}\\
\mathbf{F}_{D}
\end{bmatrix}=\mathbf{F}.
\]

\end{proof}
While the determinant of $\mathbf{V}$ or the determinant of $\mathbf{U}$
is needed to compute the determinant of $\mathbf{F}$, a major problem
is that we do not know $\mathbf{U}_{L}$ or $\mathbf{V}_{D}$, which
may not be efficiently computed due to their large degrees. This means
we need to compute the determinant of $\mathbf{V}$ or $\mathbf{U}$
without knowing the complete matrix $\mathbf{V}$ or $\mathbf{U}$.
The following lemma shows how this can be done using just $\mathbf{U}_{R}$
and $\mathbf{V}_{U}$, which are obtained from the computation of the column basis $\mathbf{G}_{1}$.
\begin{lem}
\label{lem:scalingToDeterminant} Let $\mathbf{U}=\left[\mathbf{U}_{L},\mathbf{U}_{R}\right]$
and $\mathbf{F}$ be as before, that is, they satisfy 
\[
\mathbf{F} \cdot \mathbf{U}=\begin{bmatrix}\mathbf{F}_{U}\\
\mathbf{F}_{D}
\end{bmatrix}\begin{bmatrix}\mathbf{U}_{L} & \mathbf{U}_{R}\end{bmatrix}=\begin{bmatrix}\mathbf{G}_{1} & 0\\
* & \mathbf{G}_{2}
\end{bmatrix}=\mathbf{G},
\]
 where the row dimension of $\mathbf{F}_{U}$, the column dimension
of $\mathbf{U}_{L}$, and the dimension of $\mathbf{G}_{1}$ are $k$.
Let $\mathbf{V}=\begin{bmatrix}\mathbf{V}_{U}\\
\mathbf{V}_{D}
\end{bmatrix}$ be the inverse of $\mathbf{U}$ with $k$ rows in $\mathbf{V}_{U}$.
If $\mathbf{U}_{L}^{*}\in\mathbb{K}\left[x\right]^{n\times k}$ is
a matrix such that $\mathbf{U}^{*}=\left[\mathbf{U}_{L}^{*},\mathbf{U}_{R}\right]$
is unimodular, then 
\[
\det\mathbf{F}~=~\frac{\det\mathbf{G}\cdot\det\left(\mathbf{V}_{U}\mathbf{U}_{L}^{*}\right)}{\det\left(\mathbf{U}^{*}\right)}.
\]
\end{lem}
\begin{proof}
Since $\det\mathbf{F}~=~\det\mathbf{G}\cdot\det\mathbf{V}$, we just
need to show that $\det\mathbf{V}=\det\left(\mathbf{V}_{U}\mathbf{U}_{L}^{*}\right)/\det\left(\mathbf{U}^{*}\right)$.
This follows from 
\begin{eqnarray*}
\det\mathbf{V}\cdot\det\mathbf{U}^{*} & = & \det\left(\mathbf{V}\cdot\mathbf{U}^{*}\right)\\
 & = & \det\left(\begin{bmatrix}\mathbf{V}_{U}\\
\mathbf{V}_{D}
\end{bmatrix}\begin{bmatrix}\mathbf{U}_{L}^{*} & \mathbf{U}_{R}\end{bmatrix}\right)\\
 & = & \det\left(\begin{bmatrix}\mathbf{V}_{U}\mathbf{U}_{L}^{*} & 0\\
* & I
\end{bmatrix}\right)\\
 & = & \det\left(\mathbf{V}_{U}\mathbf{U}_{L}^{*}\right).
\end{eqnarray*}

\end{proof}
\prettyref{lem:scalingToDeterminant} shows that the determinant of
$\mathbf{V}$ can be computed using $\mathbf{V}_{U}$, $\mathbf{U}_{R}$, and
a unimodular completion $\mathbf{U}^*$ of $\mathbf{U}_{R}$.
In fact, this can be made more efficient still by noticing that the
higher degree parts do not affect the computation.
\begin{lem}
\label{lem:determinantOfUnimodular}If $\mathbf{U}\in\mathbb{K}\left[x\right]^{n\times n}$
is unimodular, then $\det\mathbf{U}=\det\left(\mathbf{U}\mod x\right)=\det\left(\mathbf{U}\left(0\right)\right)$.\end{lem}
\begin{proof}
Note that $\det\left(\mathbf{U}\left(\alpha\right)\right)=\left(\det\mathbf{U}\right)\left(\alpha\right)$
for any $\alpha\in\mathbb{K}$, that is, the result is the same whether
we do evaluation before or after computing the determinant. Taking
$\alpha=0$, we have 
\[
\det\left(\mathbf{U}\mod x\right)=\det\left(\mathbf{U}\left(0\right)\right)=\left(\det\mathbf{U}\right)\left(0\right)=\det\left(\mathbf{U}\right)\mod x=\det\mathbf{U}.
\]

\end{proof}
Lemma \ref{lem:determinantOfUnimodular} allows us to use just the degree zero coefficient matrices in the computation. Hence \prettyref{lem:scalingToDeterminant} can be improved as follows.
\begin{lem}
\label{lem:scalingToDeterminantSimplified} Let $\mathbf{F}$, $\mathbf{U}=\left[\mathbf{U}_{L},\mathbf{U}_{R}\right]$,
and $\mathbf{V}=\begin{bmatrix}\mathbf{V}_{U}\\
\mathbf{V}_{D}
\end{bmatrix}$ be as before. Let $U_{R}=\mathbf{U}_{R}\mod x$ and $V_{U}=\mathbf{V}_{U}\mod x$
be the constant matrices of $\mathbf{U}_{R}$ and $\mathbf{V}_{U}$,
respectively. If $U_{L}^{*}\in\mathbb{K}^{n\times*}$ is a matrix
such that $U^{*}=\left[U_{L}^{*},U_{R}\right]$ is unimodular, then
\[
\det\mathbf{F}~=~\frac{\det\mathbf{G}\cdot\det\left(V_{U}U_{L}^{*}\right)}{\det\left(U^{*}\right)}.
\]
 \end{lem}
\begin{proof}
\prettyref{lem:determinantOfUnimodular} implies that $\det\mathbf{V}=\det V$
and $\det\mathbf{U}^{*}=\det U^{*}$. These can then be substituted
in the proof of \prettyref{lem:scalingToDeterminant} to obtain the
result.\end{proof}
\begin{rem}
\prettyref{lem:scalingToDeterminantSimplified} requires us to compute
$U_{L}^{*}\in\mathbb{K}^{n\times *}$ a matrix such that $U^{*}=\left[U_{L}^{*},U_{R}\right]$
is unimodular. This can be obtained from the unimodular matrix
that transforms $V_{U}$ to its reduced column echelon form computed
using the Gauss Jordan transform algorithm from \citep{storjohann:phd2000}
with a cost of $O\left(nm^{\omega-1}\right)$ where $m$ is the column dimension of $U_{L}^{*}$. 
\end{rem}
We now have all the ingredients needed for computing the determinant
of $\mathbf{F}$. A recursive algorithm is given in Algorithm \prettyref{alg:determinant},
which computes the determinant of $\mathbf{F}$ as the product of
the determinant of $\mathbf{V}$ and the determinant of $\mathbf{G}$.
The determinant of $\mathbf{G}$ is computed by recursively computing
the determinants of its diagonal blocks $\mathbf{G}_{1}$ and $\mathbf{G}_{2}$.

\begin{algorithm}[t]
\caption{$\determinant(\mathbf{F})$}
\label{alg:determinant}

\begin{algorithmic}[1]
\REQUIRE{$\mathbf{F}\in\mathbb{K}\left[x\right]^{n\times n}$,  nonsingular.
}

\ENSURE{the determinant of $\mathbf{F}$.}

\STATE{$\begin{bmatrix}\mathbf{F}_{U}\\
\mathbf{F}_{D}
\end{bmatrix}:=\mathbf{F}$, with $\mathbf{F}_{U}$ consists of the top $\left\lceil n/2\right\rceil $
rows of $\mathbf{F}$;}

\STATE{\textbf{if }$n=1$ \textbf{then} \textbf{return} $\mathbf{F}$; \textbf{endif};}

\STATE{$\mathbf{G}_{1},\mathbf{U}_{R},\mathbf{V}_{U}:=\colBasis(\mathbf{F}_{U})$;
\\ \hspace{0.15in} {\bf Note:} Here $\colBasis()$ also returns the kernel basis $\mathbf{U}_{R}$ and 
the
right factor $\mathbf{V}_{U}$
\\ \hspace{0.3in}
 it computed in addition to the column basis $\mathbf{G}_{1}$.}

\STATE{$\mathbf{G}_{2}:=\mathbf{F}_{D}\mathbf{U}_{R}$;}

\STATE{$U_{R}:=\mathbf{U}_{R}\mod x$; $V_{U}:=\mathbf{V}_{U}\mod x$;}

\STATE{
Compute $U_{L}^{*}\in\mathbb{K}^{n\times k}$ , a matrix that makes
 $U^{*}=\left[U_{L}^{*},U_{R}\right]$ unimodular;
}

\STATE{$d_{V}:=\det\left(V_{U}U_{L}^{*}\right)/\det(U^{*};$}

\STATE{$\mathbf{d}_{G}:=\determinant(\mathbf{G}_{1})\cdot\determinant(\mathbf{G}_{2});$}

\STATE{\textbf{return} $d_{V}\cdot\mathbf{d}_{G}$;}
\end{algorithmic}
\end{algorithm}

\begin{example}
In order to see correctness of the algorithm, let
\[
\mathbf{F}=\left[\begin{array}{rcccc}
x & -{x}^{3} & -2\,{x}^{4} & 2x & -{x}^{2}\\
\noalign{\medskip}1 & -1 & -2\, x & 2 & -x\\
\noalign{\medskip}-3 & 3\,{x}^{2}+x & 2\,{x}^{2} & -\,{x}^{4}+1 & 3\, x\\
\noalign{\medskip}0 & 1 & {x}^{2}+2\, x-2 & \,{x}^{3}+2x-2 & 0\\
\noalign{\medskip}1 & -{x}^{2}+2 & -2\,{x}^{3}-3\, x+3 & 2x+2 & 0
\end{array}\right]
\]
 working over $\mathbb{Z}_{7}[x]$. If $\mathbf{F}_{U}$ denotes the
top three rows of $\mathbf{F}$ then a column basis 
\[
\mathbf{G}_{1}=\left[\begin{array}{rcr}
x & -{x}^{3} & -2\,{x}^{4}\\
\noalign{\medskip}1 & -1 & -2\, x\\
\noalign{\medskip}-3 & 3\,{x}^{2}+x & 2\,{x}^{2}
\end{array}\right]
\]
 and minimal kernel basis 
\[
\mathbf{U}_{R}=\left[\begin{array}{rc}
-1 & x\\
-x^{2} & 0\\
-3x & 0\\
-3 & 0\\
0 & 1
\end{array}\right]
\]
for $\mathbf{F}_{U}$ were given in Example \ref{ex:example1}. The computation of the
column basis also gives the right factor 
\[
\mathbf{V}_{U}=\left[\begin{array}{cccrr}
1 & 0 & 0 & 2 & -x\\
0 & 1 & 0 & 2x^{2} & 0\\
0 & 0 & 1 & -x & 0
\end{array}\right].
\]
The constant term matrices are then 
\[
{U}_{R}=\left[\begin{array}{rc}
-1 & 0\\
0 & 0\\
0 & 0\\
-3 & 0\\
0 & 1
\end{array}\right]~\mbox{and}~~{V}_{U}=\left[\begin{array}{cccrr}
1 & 0 & 0 & 2 & 0\\
0 & 1 & 0 & 0 & 0\\
0 & 0 & 1 & 0 & 0
\end{array}\right]
\]
with Gaussian-Jordan used to fnd a unimodular completion of $U_{R}$ as  
\[
U_{L}^{*}=\left[\begin{array}{ccc}
1 & 0 & 0\\
0 & 1 & 0\\
0 & 0 & 1\\
0 & 0 & 0\\
0 & 0 & 0
\end{array}\right].
\]
 The determinant of $\mathbf{V}$ is then computed as 
\[
\frac{\det\left(V_{U}U_{L}^{*}\right)}{\det\left(U^{*}\right)}=-\frac{1}{3}=2.
\]
Recursively computing the determinant of $\mathbf{G}_{1}$ and $\mathbf{G}_{2}$ gives $\det\mathbf{G}_{1}=x^{6}-x^{4}$ and $\det\mathbf{G}_{2}=x^{4}-x$. 
Accumulating the above gives the determinant of $\mathbf{F}$ as
\[
\det\mathbf{F}=\det\mathbf{V}\cdot\det\mathbf{G}_{1}\cdot\det\mathbf{G}_{2}=2\left(x^{6}-x^{4}\right)\left(x^{4}-x\right)=2x^{10}-2x^{8}-2x^{7}+2x^{5}.
\]
\qed \end{example}

\subsection{Computational cost}
\begin{thm}
\label{thm:diagonalCost} Algorithm \prettyref{alg:determinant} 
costs $\bigO\left(n^{\omega}s\right)$ field operations to compute
the determinant of a nonsingular matrix $\mathbf{F}\in\mathbb{K}\left[x\right]^{n\times n}$,
where $s$ is the average column degree of $\mathbf{F}$. \end{thm}
\begin{proof}
From \prettyref{lem:firstDiagonalBlock} and \prettyref{lem:secondDiagonalBlock}
the computation of the two diagonal blocks $\mathbf{G}_{1}$ and $\mathbf{G}_{2}$
costs $O^{\sim}\left(n^{\omega}s\right)$ field operations, which
dominates the cost of the other operations in the algorithm. 

Now consider the cost of the algorithm on a subproblem in the recursive computation. If we let the cost
be $g(m)$ for a subproblem whose input matrix has dimension $m$, by \prettyref{lem:firstDiagonalBlock}
and \prettyref{lem:secondDiagonalBlock} the sum of the column degrees
of the input matrix is still bounded by $ns$, but the average column
degree is now bounded by $ns/m$. The cost on the subproblem is then
\begin{eqnarray*}
g(m) & \in & O^{\sim}(m^{\omega}\left(ns/m\right))+g(\left\lceil m/2\right\rceil )+g(\left\lfloor m/2\right\rfloor )\\
 & \subset & O^{\sim}(m^{\omega-1}ns)+2g(\left\lceil m/2\right\rceil )\\
 & \subset & O^{\sim}(m^{\omega-1}ns).
\end{eqnarray*}
 The cost on the original problem when the dimension $m=n$ is therefore
$O^{\sim}\left(n^{\omega}s\right)$.
\end{proof}

\section{Conclusion\label{sec:Future-Research}}

In this paper we have given a new, fast, deterministic algorithm for computing
the determinant of a nonsingular polynomial matrix. Our method relies on the efficient, deterministic computation of the diagonal elements of a triangularization of the input matrix. This in turn relies on recent efficient methods \cite{za2013,za2012} for computing shifted minimal kernel and column bases of polynomial matrices.

In a future report we will show how our triangularization technique results in a fast, deterministic algorithm for finding a Hermite normal form. Other directions of interest include making use of the diagonalization procedures in domains such as matrices of differential or, more generally, of Ore operators, particularly for computing normal forms. Partial results had been reported in \cite{davies-cheng:2006}, at least in the case of Popov normal forms. 
Similarly we are interested in applying our block elimination techniques using kernal bases to computing the Dieudonn\'e determinant and quasideterminant of matrices for Ore polynomial rings. These are the two main generalizations of determinants for matrices of noncommutative polynomials.  Degree bounds for these noncommutative determinants have been used in \cite{markkim} for modular computation of normal forms.

\bibliographystyle{plain}


\end{document}